\documentclass{ws-fnl}

\usepackage{amsmath,amssymb,bm}
\usepackage{graphicx}
\usepackage{booktabs}
\usepackage{xcolor}
\usepackage[verbose]{hyperref}
\hypersetup{colorlinks=false,allbordercolors=blue,pdfborderstyle={/S/U/W 1}}
\makeatletter
\@ifundefined{c@prop}{\newtheorem{prop}{Proposition}}{}
\@ifundefined{c@lem}{}{}
\@ifundefined{c@rem}{\newtheorem{rem}[prop]{Remark}}{}
\@ifundefined{proof}{%
  \newenvironment{proof}{\par\noindent{\itshape Proof.}\ }%
                        {\hfill$\square$\par\medskip}}{}
\makeatother

\newcommand{\dd}{\mathrm{d}}

\begin{document}

\markboth{M. Bianucci, R. Mannella}
{It\^o versus Stratonovich for state-dependent shot noise}

\catchline{}{}{}{}{}

\title{It\^o versus Stratonovich for shot noise with state-dependent
amplitude: why the answer is Marcus}

\author{Marco Bianucci}
\address{Istituto di Scienze Marine, Consiglio Nazionale delle Ricerche
(ISMAR--CNR), 19032 Lerici (SP), Italy}

\author{Riccardo Mannella}
\address{Dipartimento di Fisica, Universit\`a di Pisa, 56100 Pisa, Italy
\\ riccardo.mannella@unipi.it}

\date{\today}
\maketitle

\begin{history}
\received{(Day Month Year)}
\accepted{(Day Month Year)}
\end{history}

\begin{abstract}
For Gaussian white noise the It\^o and Stratonovich prescriptions differ
by $O(\dd t)$ per time step, which accumulates into the familiar spurious
drift $\tfrac12 D\,II'$. We show that for renewal shot noise with
state-dependent gain, $\dot x=-C(x)-I(x)\xi[t]$ with $\xi[t]$ a train of
Dirac spikes of finite random amplitude, the two prescriptions differ by
a finite amount at every event, and that where $I$ vanishes the
difference can become qualitative. The Marcus kick, the jump-process
counterpart of the Stratonovich prescription, is the flow map of
$\dd x/\dd s=I(x)$: no kick can cross a zero $x^\ast$ of $I$, for any
jump and waiting-time statistics. The It\^o kick is only the tangent map
of that flow. For $I(x)=1+\beta x$ it crosses $x^\ast=-1/\beta$ if and
only if $\xi\beta>1$ and collapses the whole line onto $x^\ast$ at
$\xi\beta=1$; below threshold it is itself a Marcus kick with a remapped
jump law. Bounded jump laws thus have a threshold, while unbounded ones
leak across the barrier for \emph{any} parameters. The discrepancy cannot
be absorbed into a renormalised drift beyond lowest order, and only the
Marcus kick is generated by a single operator, which makes the master
equation of Bianucci, Bologna and Mannella (2026) intrinsically Marcus.
Numerically, the standard Heun kick is a second-order truncation of the
Marcus map that preserves the barrier only because its order is even; it
shortens every jump, biasing second moments in proportion to the squared
jump amplitude and the jump kurtosis, and depleting the region near the
barrier. We give the exact-flow integrator that replaces it.
\end{abstract}

\keywords{Shot noise; renewal processes; It\^o--Stratonovich dilemma;
Marcus canonical equation;
multiplicative noise; continuous time random walks.}

\section{Introduction}

The It\^o--Stratonovich dilemma is among the oldest and most discussed
questions in the theory of multiplicative noise~\cite{VanKampen1981,Risken1996}.
For the Langevin equation $\dot x=-C(x)-I(x)\zeta(t)$ driven by Gaussian
white noise with $\langle\zeta(t)\zeta(t')\rangle=D\,\delta(t-t')$, the two
prescriptions are related by the spurious drift $\tfrac12 D\,I(x)I'(x)$.
The drift is finite, but it builds up from a difference of order $\dd t$
per time step: at each step the two readings differ only infinitesimally,
because the increments of the Wiener process are infinitesimal. Which of the two is physically
correct is then a question about the vanishing correlation time of the
real physical fluctuation, and the standard answer -- Stratonovich when
the white noise is the idealisation of a coloured one, It\^o when the
increments are genuinely non-anticipating -- is by now textbook
material. The state of that discussion was surveyed in this journal in
Ref.~\refcite{Mannella2012}; the present paper may be read as a sequel to
it, asking what survives when the idealised white noise is replaced by a
train of spikes of finite amplitude, i.e. a renewal
train of spikes,
\begin{equation}
  \xi(t)=\sum_{q=0}^{\infty}\xi_q\,
  \delta\!\left(t-t_0-\sum_{k=0}^{q}\theta_k\right),
  \label{eq:spike}
\end{equation}
with waiting times $\theta_k$ drawn from $\psi(\theta)$ and amplitudes
$\xi_q$ from $p(\xi)$. Such shot-noise forcing is the standard driving of
continuous-time random walk (CTRW) models and appears among others, with
state-dependent gain, in synaptic input in neuroscience, in impulsive
forcing of slow climate modes, and in intermittent crack
growth (see Ref.~\refcite{BianucciMain} for more examples and related
literature). Here the amplitude $\xi_q$ of a single event
is a \emph{finite} random number, not an infinitesimal increment. The
ambiguity in evaluating $I(x)$ ``during'' a spike therefore does not
produce a small correction: it produces two different, finitely
separated, maps of the line onto itself.

The main purpose of this paper is to make that statement precise and to show
that its consequences are not merely quantitative. For noise with jumps, the natural counterpart of the Stratonovich prescription is not the Stratonovich integral itself but the canonical equation introduced by Marcus~\cite{Marcus1978,Marcus1981}, in which each jump is realised by flowing along the vector field of the noise for a virtual time equal to the jump amplitude. We call it the Marcus prescription, and discuss its relation to the Stratonovich integral in Section~\ref{sec:marcus}.
Our main result is
that the Marcus prescription  endows the kick with the structure of a flow -- hence with invertibility, a group law, and invariant sets --
while the It\^o prescription retains only its linearisation and in general loses all three.  When $I$ has a zero, the difference becomes a difference in the
support of the probability density, which is the most visible thing a
stochastic model can have. A second result of this paper is to pinpoint the appropriate numerical scheme to accurately integrate the SDE, under the Marcus prescription in the presence of delta kicks.

\section{The two kick maps}

\subsection{The Marcus (flow-map) kick}

Consider the SDE
\begin{equation}
  \dot x = -C(x) - I(x)\,\xi[t]
  \label{eq:sde}
\end{equation}
and regularise the $q$-th spike as a smooth pulse of unit area and width
$\epsilon$, $\delta_\epsilon(t-t_q)\to\delta(t-t_q)$. During the pulse
the drift $C$, being bounded, contributes $O(\epsilon)$ and may be
dropped, so that
\begin{equation}
  \frac{\dd x}{\dd t}=-\xi_q\,I(x)\,\delta_\epsilon(t-t_q).
\end{equation}
Introducing the virtual time
$s(t)=-\xi_q\int_{t_q-\epsilon}^{t}\delta_\epsilon(u-t_q)\dd u$,
which runs from $0$ to $-\xi_q$ irrespective of the pulse shape,
the equation becomes autonomous,
\begin{equation}
  \frac{\dd x}{\dd s}=I(x),
  \label{eq:flow}
\end{equation}
and the kick is the time-$(-\xi_q)$ map of its flow:
\begin{equation}
  x(t_q^{+}) = \phi_{-\xi_q}\!\left(x(t_q^{-})\right),
  \qquad
  \phi_s(x)=Y^{-1}\!\left(Y(x)+s\right),
  \qquad
  Y(x):=\int^{x}\frac{\dd z}{I(z)} 
  \label{eq:stratmap}
\end{equation}
The simplest case when $I(x)$ has a zero is $I(x)=1+\beta x$: then the map reads
\begin{equation}
  u \mapsto u\,e^{-\xi\beta},
  \qquad u:=x-x^\ast,\quad x^\ast=-\tfrac1\beta,
  \label{eq:linearcase}
\end{equation}

\begin{prop}[Shape independence]
\label{prop:shape}
The map \eqref{eq:stratmap} is independent of the regularising pulse
profile $\delta_\epsilon$ and of the limit $\epsilon\to0$, provided $\delta_\epsilon$ has unit area, $C$ is bounded on the relevant
interval, and the flow of Eq.~\eqref{eq:flow} is complete (for instance, $I$
of at most linear growth).
\end{prop}
\begin{proof}
Since Eq.~\eqref{eq:flow} is autonomous, $x(t)=\phi_{s(t)}\big(x(t_q-\epsilon)\big)$
for any path $s(t)$, monotone or not, and only its end value $s=-\xi_q$
matters.  The
neglected drift contributes $O(\epsilon\sup|C|)$.
\end{proof}

Proposition~\ref{prop:shape} is the exact analogue, for spikes of finite
area, of the Wong--Zakai theorem~\cite{WongZakai1965}. It is what makes
the Marcus prescription the natural one within a stochastic
Liouville description along single trajectories: the continuity equation
holds for each realisation, and the kick is a mass-conserving transport along Eq.~\ref{eq:flow}.

\subsection{The It\^o kick}

The It\^o (pre-point) prescription evaluates $I$ at the state immediately
before the event,
\begin{equation}
  x(t_q^{+}) = x(t_q^{-}) - \xi_q\,I\!\left(x(t_q^{-})\right)
  \;=:\; m_{\xi_q}\!\left(x(t_q^{-})\right).
  \label{eq:itomap}
\end{equation}
Expanding \eqref{eq:stratmap} in powers of the jump amplitude,
\begin{equation}
  \phi_{-\xi}(x) = x - \xi I + \frac{\xi^2}{2}II'
  - \frac{\xi^3}{6}I\left(II''+I'^2\right)+O(\xi^4),
  \label{eq:expansion}
\end{equation}
shows that \eqref{eq:itomap} reproduces only the first-order term. The
familiar Stratonovich--It\^o drift $\tfrac12 II'$ is recognisable as the
$O(\xi^2)$ coefficient; what is new here is that $\xi$ is a finite random
variable, so the series is not a small correction but an expansion that
may not even converge usefully. 

For $I(x)=1+\beta x$, the It\^o map reads
\[u\mapsto u(1-\xi\beta) \] 
with $u$ as in Eq.~\eqref{eq:linearcase}: the truncation at first order of the exponential, a sensible approximation only for $|\xi|\beta\ll1$..

\begin{rem}[Post-point prescription]
The post-point prescription~\cite{klimontovich1990,Ao2004Potential} 
$x^{+}=x^{-}-\xi I(x^{+})$ gives $u\mapsto u/(1+\xi\beta)$, i.e.\ the
$[0/1]$ Pad\'e approximant of $e^{-\xi\beta}$. It has a mirrored threshold: the map is undefined at $\xi\beta=-1$ and,
for $\xi\beta<-1$, reverses orientation and carries trajectories across
$x^\ast$.
\end{rem}


Reference~\refcite{BianucciMain} notes in passing that calculi other than
Stratonovich ``can be readily studied, both numerically and
theoretically, renormalizing $C(x)$'', citing the white-noise analysis of
Ref.~\refcite{Mannella2012}. For Gaussian white noise driving 
this is exact: the
whole It\^o--Stratonovich difference is the drift
$\tfrac12 D\,I I'$ and nothing else. For spike noise of finite amplitude
it holds only to lowest order. Indeed, from \eqref{eq:expansion} the two
kick maps differ by
\begin{equation}
  \phi_{-\xi}(x)-m_\xi(x)
  = \frac{\xi^2}{2}I I' - \frac{\xi^3}{6}I\left(II''+I'^2\right)+O(\xi^4),
  \label{eq:mapdiff}
\end{equation}
and only the first term, averaged over $p(\xi)$ and over the event 
rate $R(t)$,
is equivalent to the drift renormalisation
$\Delta C=-\tfrac{1}{2}R(t)\overline{\xi^2}\,I I'$. The remainder is a
genuine difference between two \emph{bounded jump operators}, and a
bounded jump operator is not a first-order differential operator: no
choice of $\Delta C$ reproduces it. The obstruction is not merely one of accuracy. Since $I(x^\ast)=0$, both
kick maps fix $x^\ast$, and every term of Eq.~\eqref{eq:mapdiff} carries a
factor $I$ and vanishes there; so does any drift $\Delta C$ assembled from
them. A renormalised drift therefore leaves the velocity at $x^\ast$
unchanged: whatever transport across $x^\ast$ the original drift $C$
produces, the renormalised one produces identically, and nothing more.
The It\^o kick, by contrast, transports probability across $x^\ast$ above
the threshold of Proposition~\ref{prop:threshold}, on top of whatever $C$
does, and no choice of $\Delta C$ can reproduce that. We return to this in
Sec.~\ref{sec:barrier}.


The maps of Eq.~\eqref{eq:stratmap} and Eq.~\eqref{eq:itomap} differ in an algebraic
property that matters for renewal dynamics.

\begin{prop}[Group law]
\label{prop:group}
The Marcus kicks form a one-parameter group,
$\phi_{s_1}\circ\phi_{s_2}=\phi_{s_1+s_2}$, so that $n$ events of
amplitudes $\xi_1,\dots,\xi_n$ occurring at coincident times are
equivalent to a single event of amplitude $\sum_i\xi_i$. The It\^o maps
do not:
$m_{\xi_1}\circ m_{\xi_2} \ne m_{\xi_1+\xi_2}$ unless $I$ is constant,
the discrepancy being $\xi_1\xi_2\,II'+O(\xi^3)$.
\end{prop}

A practical consequence is that the sub-stepping used in SDE integrators,
where several events falling in a single time step are accumulated into
one kick of amplitude $A=\sum_i\xi_i$, is exact under Marcus and
introduces an $O(\xi^2)$ error per coincidence under It\^o.

Proposition~\ref{prop:group} is not, however, what makes an exact master
equation possible. Let $\mathcal{L}_0P=\partial_x(CP)$ be the drift
operator, $\mathcal{K}=\int p(\xi)\,\mathcal{K}_\xi\,\dd\xi$ the averaged
kick operator on densities, with $\mathcal{K}_\xi$ the push-forward under
either kick map, and $\Psi(t)=\int_t^\infty\psi(t')\,\dd t'$ the survival function.
Conditionally on the event times the density evolves by alternating
$e^{\mathcal{L}_0\theta}$ and $\mathcal{K}$, and for independent,
identically distributed waiting times the renewal series sums in Laplace
space,
\begin{equation}
  \hat P(s)=\hat\Phi_\Psi(s)\left[1-\mathcal{K}\,\hat\Phi_\psi(s)\right]^{-1}P_0,
  \qquad
  \hat\Phi_f(s)=\int_0^\infty f(t)\,e^{-st}\,e^{\mathcal{L}_0t}\,\dd t ,
  \label{eq:renewal}
\end{equation}
for \emph{any} kick operator, whether or not the kicks commute and whether
or not the kick map is invertible. An exact, non-local evolution equation
therefore exists under both prescriptions. What distinguishes Marcus is
its form. The Marcus push-forward is $\mathcal{K}_\xi=e^{\xi\partial_xI(x)}$,
the exponential of a single operator independent of $\xi$, so that
$\mathcal{K}=\hat p(-i\partial_xI(x))$: the jump law enters only through
its characteristic function, and its moments multiply powers of one and
the same operator. This is the structure on which the cumulant expansion
and the mean-field local equation of Ref.~\refcite{BianucciMain} rest. For a general
$I$ the It\^o push-forwards $\mathcal{M}_\xi$ admit no such common
generator, and the averaged It\^o kick is not a function of one operator
(the linear case is special; see Remark~\ref{rem:remap}).
 
The block structure of the $n$-time correlation functions of (1), in
which a block of $m$ coincident times contributes $\overline{\xi^m}$, is
a separate matter: it arises because the same spike is counted $m$ times,
and it is a property of the noise alone, independent of any kick map.
The single-generator form of $\mathcal{K}$ is what allows it to be carried
over to the evolution of $P$. In this precise sense the exact theory of
Ref.~\refcite{BianucciMain} is a Marcus theory.

\subsection{Relation to the Stratonovich integral}
\label{sec:marcus}

The map \eqref{eq:stratmap} is not a notational variant of the
Stratonovich integral. For driving processes with jumps the
Fisk--Stratonovich integral and the flow-map prescription are different
objects, and the one defined by \eqref{eq:stratmap} is the \emph{Marcus
canonical} (or simply canonical) equation, introduced by
Marcus~\cite{Marcus1978,Marcus1981} 
and developed in generality by Kurtz,
Pardoux and Protter~\cite{KPP1995}.

The reason the distinction matters is that the two properties for which
physicists prefer Stratonovich in the Brownian case -- the ordinary chain
rule, and the fact that it arises as the small-correlation-time limit of
a smooth approximation -- are retained in the jump case by Marcus and not
by the Stratonovich integral proper. This comparison is made explicitly
by Chechkin and Pavlyukevich~\cite{ChechkinPavlyukevich2014}, who derive
the Marcus equation as the limit of Ornstein--Uhlenbeck coloured jump
noise of vanishing correlation time. Proposition~\ref{prop:shape} is the
one-dimensional specialisation of that statement to a single spike of
finite area, and Proposition~\ref{prop:group} is the semigroup property
of the Marcus flow. Neither is claimed here as new; they are recalled
because the rest of the paper rests on them.

What is specific to the present work is what the flow structure implies
once $I$ has a zero -- the one-way, repelling barrier of Proposition~\ref{prop:barrier} and the
equilibrium law Eq.~\eqref{eq:wkb_result} --
together with the sharp threshold, the collapse and the orientation
reversal of the It\^o map (Proposition~\ref{prop:threshold}), the
resulting bounded/unbounded classification of Table~\ref{tab:crossing},
and the analysis of the Heun kick in Sec.~\ref{sec:numerics}.

\begin{rem}[Relation to the physicists' usage]
Much of the physics literature on delta-correlated noise, including
Ref.~\refcite{BianucciMain}, uses ``Stratonovich'' to mean the
smooth-pulse (Wong--Zakai) limit. Under that reading the two names
denote the same object here, and no result of
Ref.~\refcite{BianucciMain} is affected. We nonetheless use ``Marcus''
throughout, because the point of this paper is precisely that the
tangent-map and flow-map prescriptions are inequivalent, and that
distinction is invisible if both are called Stratonovich.
\end{rem}

\subsection{The series and its resummation}\label{sec:series}

For Poisson white noise, Di~Paola and Falsone showed that the correction
relating the It\^o and the Stratonovich prescriptions is not a single
term but an infinite series in the jump
amplitude~\cite{DiPaolaFalsone1993a,DiPaolaFalsone1993b}. That is the
series side of our Eq.~\eqref{eq:mapdiff}, and it already contains the
statement that no finite drift renormalisation can account for the
difference.

Two things are added here. The first is the resummation: the series is
the Taylor expansion of the flow map $\phi_{-\xi}$, so the correction has
a closed form, and for $I=1+\beta x$ that closed form is the elementary
factor $e^{-\xi\beta}$ of Eq.~\eqref{eq:linearcase}. 
The second, and the more consequential, concerns the barrier. Whether
$x^\ast$ is preserved depends on the sign of a multiplicative factor,
not on the size of a correction, and it is not settled order by order.
For $I=1+\beta x$ the truncation of $e^{-z}$, $z=\xi\beta$, at order $n$
is $T_n(z)=\sum_{k=0}^{n}(-z)^k/k!$. For odd $n$ it has exactly one real
root, which is positive: $z=1$ at first order -- the It\^o threshold of
Proposition~\ref{prop:threshold} -- $z\simeq1.596$ at third order,
$z\simeq2.181$ at fifth. For even $n$ it has none, like the exact factor.
(Since $T_n'=-T_{n-1}$, at a critical point of an even-order truncation
$T_n(z_0)=z_0^n/n!>0$; an odd-order one is therefore strictly
decreasing, and positive for $z\le0$.) Truncations of odd order thus all
cross the barrier, each at a different threshold, while those of even
order preserve it, the second-order one included. Preservation of the
barrier alternates with the parity of the truncation instead of
converging to the exact behaviour; only the resummed map guarantees it.

\begin{rem}[The It\^o master equation]
Equation \eqref{eq:MEito} is an instance of the generalised
Fokker--Planck equation derived, in the It\^o interpretation and for
multiplicative noise with arbitrary increment distribution, by Denisov,
Horsthemke and H\"anggi~\cite{Denisov2009}. The integral defined by
Kanazawa, Sagawa and Hayakawa~\cite{Kanazawa2012} as the white-noise limit
of coloured non-Gaussian noise is neither It\^o nor Fisk--Stratonovich
and is closely related to the flow map used here.
\end{rem} 

\section{Zeros of \texorpdfstring{$I$}{I}: a barrier in one prescription, a threshold in the other}
\label{sec:barrier}

The main result of the paper is most clearly stated against the
white-noise case. With Gaussian white noise, a zero $x^\ast$ of $I$
behaves in the same way under both prescriptions: the paths are
continuous, the noise vanishes at $x^\ast$, and so does the spurious drift
$\tfrac12 D\,II'$ that distinguishes the two readings. In both, the
velocity at $x^\ast$ is $-C(x^\ast)$, and the prescriptions agree on which
side of $x^\ast$ a trajectory can reach; they differ quantitatively, not
in the topology of the dynamics. With spikes of finite amplitude this is
no longer so. A finite jump can carry a trajectory across a point that no
continuous path can cross, and whether it does depends on the
prescription: the flow-map kick never does, the It\^o kick does above a
sharp threshold. This is the concrete content of the statement that the
difference is finite rather than infinitesimal.

\subsection{Marcus: invariant sectors}

\begin{prop}[Barrier to kicks]
\label{prop:barrier}
Let $I$ be locally Lipschitz, with a complete flow, and let $I(x^\ast)=0$.
Then $|Y(x)|\to\infty$ as $x\to x^\ast$, and consequently
$\phi_s$ maps each maximal interval on which $I$ has constant sign into
itself, for every finite $s$. Under the Marcus prescription no kick can therefore carry a trajectory
across $x^\ast$, whatever the jump amplitudes, the jump distribution, or
the waiting-time statistics.
\end{prop}
\begin{proof}
Local Lipschitz continuity gives $|I(z)|\le L|z-x^\ast|$ near $x^\ast$, so
$\int\dd z/|I(z)|\ge\int\dd z/(L|z-x^\ast|)$ and $Y$ diverges at least
logarithmically. Adding a finite $s$ to $Y$ in
\eqref{eq:stratmap} cannot leave the image of the interval.
\end{proof}

Proposition~\ref{prop:barrier} concerns the kicks. The drift, being a
continuous vector field, does cross $x^\ast$ whenever $C(x^\ast)\neq0$,
but only in the direction of its velocity there. For the case of interest,
$C(x)=\gamma x$, that velocity is $-C(x^\ast)=\gamma/\beta>0$, directed
into the sector $u>0$. That sector is therefore forward-invariant under
the full Marcus dynamics, whereas the sector $u<0$ is not: a trajectory
started below $x^\ast$ is carried across it by the drift and never
returns. The barrier is one-way. Since all our initial conditions lie
above $x^\ast$, the Marcus and Heun densities carry no mass below it at
any time, and any mass found there under It\^o can only have been placed
by a kick.

The Marcus master equation already acknowledges this. For $I(x)=1+\beta x$ and dichotomous kicks of amplitude $a$, 
the operators $e^{\pm a\partial_xI}$
act as dilations in $y=1+\beta x$. For exponential waiting times, $\psi(t)=e^{-t/\tau}/\tau$, a 
drift $C(x)=\gamma x$, and with $z\equiv a\beta$,
the stationary density $P(y)$ satisfies
\begin{equation}
0=\partial_y\!\left[\gamma(y-1)P(y)\right]
+\frac{1}{\tau}\left[\tfrac12 e^{-z}P\!\left(e^{-z}y\right)
+\tfrac12 e^{z}P\!\left(e^{z}y\right)-P(y)\right].
\label{eq:master_y}
\end{equation}
For $y\to0$, the dominant balance is between the drift and the inflow
from larger $y$:
\begin{equation}
\gamma P'(y)\simeq\frac{e^{z}}{2\tau}\,P\!\left(e^{z}y\right).
\label{eq:balance}
\end{equation}
We set $\ell\equiv\ln(1/y)$, $P=e^{S(\ell)}$, and $\nu(\ell)=-S'(\ell)$, so that $\nu$ is the local
exponent $\mathrm{d}\ln P/\mathrm{d}\ln y$. Expanding
$S(\ell-z)\simeq S(\ell)+z\,\nu-\tfrac{z^{2}}{2}\nu'$ with $\nu'\simeq 1/z$,
Eq.~\eqref{eq:balance} becomes
\begin{equation}
z\,\nu=\ell+\ln \nu+\ln(2\gamma\tau)-\frac{z}{2}.
\label{eq:wkb_implicit}
\end{equation}
Iterating once, with $L\equiv\ln(2\gamma\tau)$ and $c\equiv L-z/2$, we get
\begin{equation}
\nu(\ell)\simeq\frac{1}{z}\left[\ell+c+\ln\frac{\ell+L}{z}\right],
\end{equation}
and integrating gives
\begin{equation}
P(y)\simeq A\exp\!\left\{-\frac{1}{z}\left[\frac{\ell^{2}}{2}+c\,\ell
+(\ell+L)\left(\ln\frac{\ell+L}{z}-1\right)\right]\right\}.
\label{eq:wkb_result}
\end{equation}
When $y \to 0^+$, since $\nu\sim\ell/z\to\infty$, the density vanishes faster than any power of $y$. This is confirmed by simulations, see Fig.~\ref{fig:wkb} and 
discussion in Sec.~\ref{sec:numerics}.
The parameter $z=a\beta$ can also be written as $z=a/|x^\ast|$: it
measures the typical jump against the distance from the origin to the
zero of $I$.

\subsection{It\^o: a sharp threshold and a collapse}

Take $I(x)=1+\beta x$, $\beta>0$, $x^\ast=-1/\beta$, and an initial
condition $x>x^\ast$, i.e.\ $u=x-x^\ast>0$. The It\^o map gives
$u\mapsto u(1-\xi\beta)$, so

\begin{prop}[It\^o crossing threshold]
\label{prop:threshold}
Under the It\^o prescription the point $x^\ast$ is crossed at an event of
amplitude $\xi$ if and only if $\xi\beta>1$. Moreover:
\begin{itemize}
\item[(i)] for $\xi\beta<1$ the map is orientation preserving and the
      sector $u>0$ is invariant, as in the Marcus case;
\item[(ii)] at $\xi\beta=1$ the map $m_\xi$ is not injective: it collapses
      the \emph{entire} real line onto the single point $x^\ast$, so that
      a single event turns any density into $\delta(x-x^\ast)$;
\item[(iii)] for $\xi\beta>1$ the map is orientation reversing and
      every point of the sector is transported across $x^\ast$.
\end{itemize}
\end{prop}
\begin{proof}
Immediate from $u\mapsto u(1-\xi\beta)$; for (ii) note that
$m_\xi(x)=x(1-\xi\beta)-\xi\equiv-1/\beta$ for all $x$ when $\xi\beta=1$.
\end{proof}

Item (ii) deserves emphasis: the It\^o prescription is not merely a
different model, it is \emph{singular} on a codimension-one set of
parameters. Nothing of the kind happens under Marcus, where
$\phi_{-\xi}$ is a diffeomorphism for every $\xi$.

\begin{rem}[Below threshold, It\^o is Marcus with a remapped jump law]
\label{rem:remap}
For $I=1+\beta x$ the It\^o maps $u\mapsto(1-\xi\beta)u$ commute with
one another, although they do not satisfy the additive group law of
Proposition~\ref{prop:group}. Whenever $\xi\beta<1$ one can write
\begin{equation}
  1-\xi\beta=e^{-\xi'\beta},\qquad \xi'=-\frac{\ln(1-\xi\beta)}{\beta},
  \label{eq:remap}
\end{equation}
so each It\^o kick coincides \emph{exactly} with the Marcus kick of
amplitude $\xi'$. If $p(\xi)$ is supported in $\xi\beta<1$, the It\^o model
is therefore a Marcus model with the jump law $p'$ obtained by
transporting $p$ through Eq.~\eqref{eq:remap}, and the exact theory of
Ref.~\refcite{BianucciMain} applies to it with $\hat p$ replaced by
$\hat p'$. The remapped law is biased even when $p$ is symmetric: since
$-\ln(1-w)=w+w^2/2+\dots$, one has
$\beta\,\overline{\xi'}=\overline{\xi^2}\beta^2/2+O(\beta^4\overline{\xi^4})>0$.
For dichotomous jumps at $z=a\beta=0.5$ the amplitudes $\pm a$ become
$+1.386\,a$ and $-0.811\,a$. Below threshold the two prescriptions thus
differ by a reparametrisation of the noise -- a quantitative difference,
although not a small one.
 
The threshold acquires a transparent meaning. As $\xi\beta\to1^-$,
$\xi'\to+\infty$: the It\^o kick becomes a Marcus kick of infinite
amplitude, which sends every point onto $x^\ast$ -- the collapse of
Proposition~\ref{prop:threshold}(ii). For $\xi\beta>1$ the factor
$1-\xi\beta$ is negative and no real $\xi'$ exists, because no flow map
reverses orientation. The topological transition is the point at which
the remapped jump law ceases to exist.
\end{rem}

\subsection{Bounded versus unbounded jump distributions}

Combining Proposition~\ref{prop:threshold} with some specific jump laws gives a clean
classification. 
\begin{table}
\begin{center}
\begin{tabular}{lcc}
\toprule
jump PDF & It\^o crossing  & per-event crossing  \\
 & condition & probability \\
\midrule
dichotomous $\xi=\pm a$ & $a\beta>1$ & $\tfrac12\,\Theta(a\beta-1)$ \\
uniform on $[-\sqrt3a,\sqrt3a]$ & $\sqrt3a\beta>1$ &
   $\tfrac12\left(1-\tfrac{1}{\sqrt3 a\beta}\right)_{\!+}$ \\
Gaussian, variance $a^2$ & always & $\tfrac12\,\mathrm{erfc}\!\left(\tfrac{1}{\sqrt2\,a\beta}\right)>0$ \\
two-sided exponential, scale $a/\sqrt{2}$
& always & $\tfrac12 e^{-\sqrt{2}/(a\beta)}>0$ \\
\bottomrule
\end{tabular}
\end{center}
\caption{It\^o crossing condition and per-event crossing probability for the jump PDF's considered.}
\label{tab:crossing}
\end{table}
All distributions of Table~\ref{tab:crossing} have zero average
and variance $a^2$, but the expected behaviour is very different.
Qualitatively:
\begin{itemize}
\item For \emph{bounded} jump distributions the two prescriptions (Marcus
vs It\^o) agree on the
topology of the dynamics below a sharp threshold in $a\beta$ and disagree
above it (see also Figs.~\ref{fig1} and~\ref{fig2}).
\item For \emph{unbounded} jump distributions they disagree for every
choice of parameters, because the barrier leaks under It\^o with finite probability per
event, while no Marcus kick can cross it
(see Fig.~\ref{fig3}).
\end{itemize}

We note that the multiplicative examples of Ref.~\refcite{BianucciMain} use
$\beta=0.5$ with dichotomous jumps of unit variance, i.e.\ $a\beta=0.5$, which is
below the threshold of Proposition~\ref{prop:threshold}. The effect
discussed here is therefore invisible in those figures, making the present paper 
particularly relevant, given that we will choose $a\beta$ on both sides of 
unity here.

\section{The two master equations}

Under Marcus the kick operator is the push-forward along the flow,
whose generator is $\partial_x I(x)$, giving the operator of
Ref.~\refcite{BianucciMain},
\begin{equation}
  \partial_t P = \partial_x C(x)P
  + R(t)\left[\hat p\!\left(-i\partial_x I(x)\right)-1\right]P,
  \qquad
  \hat p\left(-i\partial_x I\right)=\int\! p(\xi)\,e^{\xi\partial_x I(x)}\dd\xi .
  \label{eq:MEstrat}
\end{equation}
Under It\^o the kick is the push-forward along $m_\xi$, which is not the
exponential of any first-order operator, and the master equation reads
\begin{equation}
  \partial_t P = \partial_x C(x)P
  + R(t)\left[\int\! p(\xi)\,\mathcal{M}_\xi\,\dd\xi-1\right]P,
  \qquad
  \left(\mathcal{M}_\xi P\right)(x)
  = \frac{P\!\left(m_\xi^{-1}(x)\right)}{\left|m_\xi'\!\left(m_\xi^{-1}(x)\right)\right|},
  \label{eq:MEito}
\end{equation}
where the push-forward is understood as a sum over preimages. For
$I=1+\beta x$ each point has a single preimage when $\xi\beta\ne1$, while
at $\xi\beta=1$ the map is constant and
$\mathcal{M}_\xi P=\delta(x-x^\ast)$ -- the analytic counterpart of
Proposition~\ref{prop:threshold}(ii). With the rate $R(t)$, Eq.~\eqref{eq:MEstrat} is the mean-field local equation of Ref.~\refcite{BianucciMain}: it is exact only for exponential waiting times, $R=1/\tau$, and Eq.~\eqref{eq:MEito} is the corresponding It\^o form. For any other waiting-time law the exact evolution, under either prescription, is the non-local Eq.~\eqref{eq:renewal}.

\begin{rem}
Both operators are Markov (positivity- and mass-preserving), so the two equations are equally ``legitimate'' as evolution
equations; the point
is that only \eqref{eq:MEstrat}  is generated by a flow -- its kick is invertible for
every $\xi$ -- and is a function of the single operator 
$\partial_xI(x)$, so that the jump law enters only through $\hat p$.
\end{rem}

\section{Numerical experiments} \label{sec:numerics}
\begin{figure}
    \centering
    \includegraphics[width=0.32\linewidth]{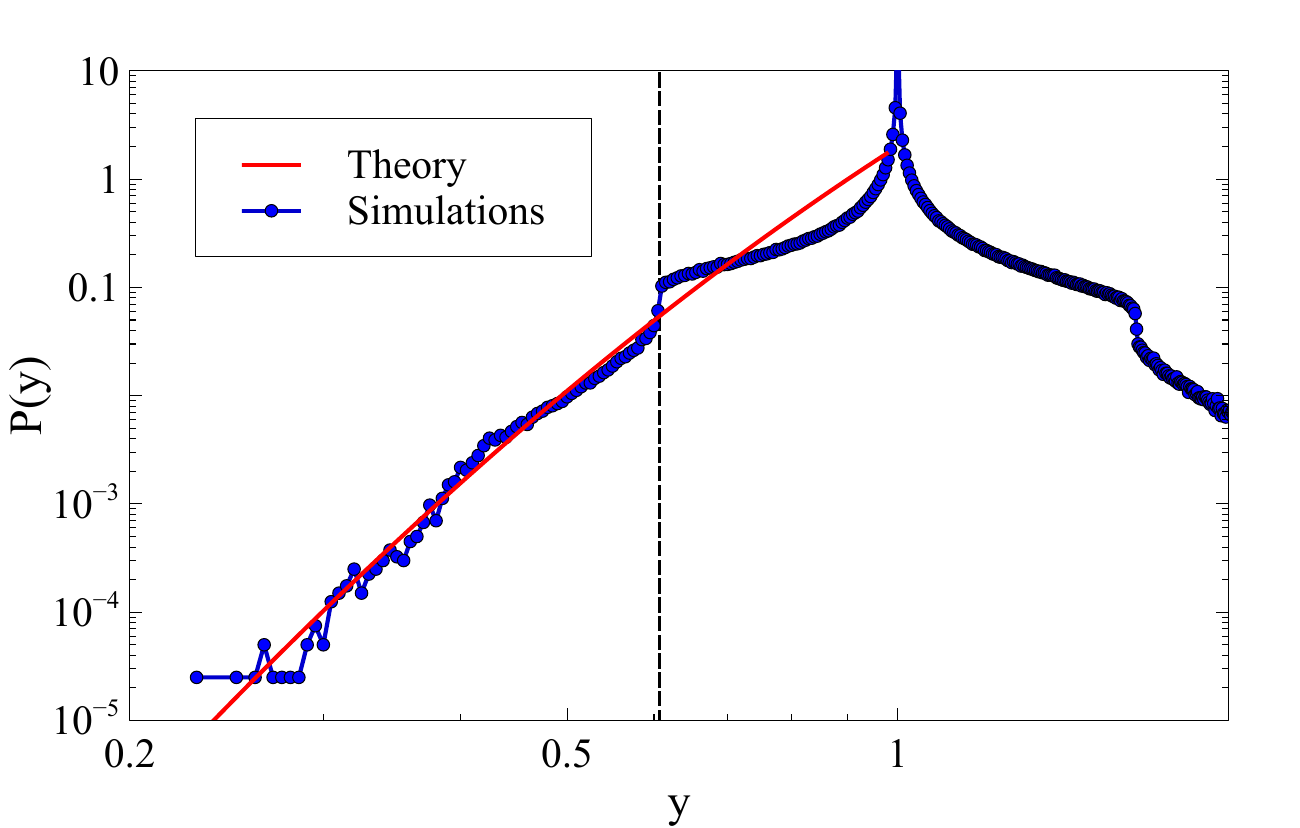}
    \includegraphics[width=0.32\linewidth]{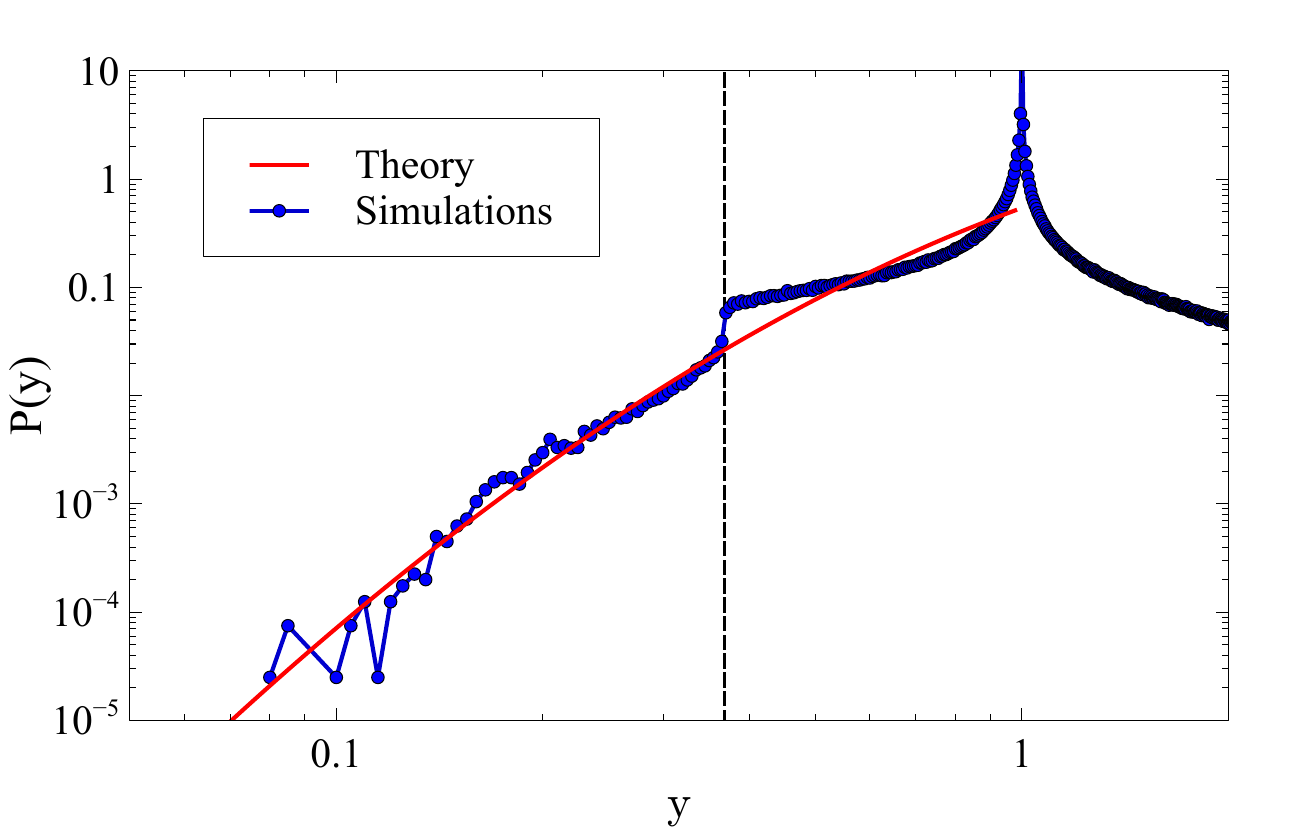}
    \includegraphics[width=0.32\linewidth]{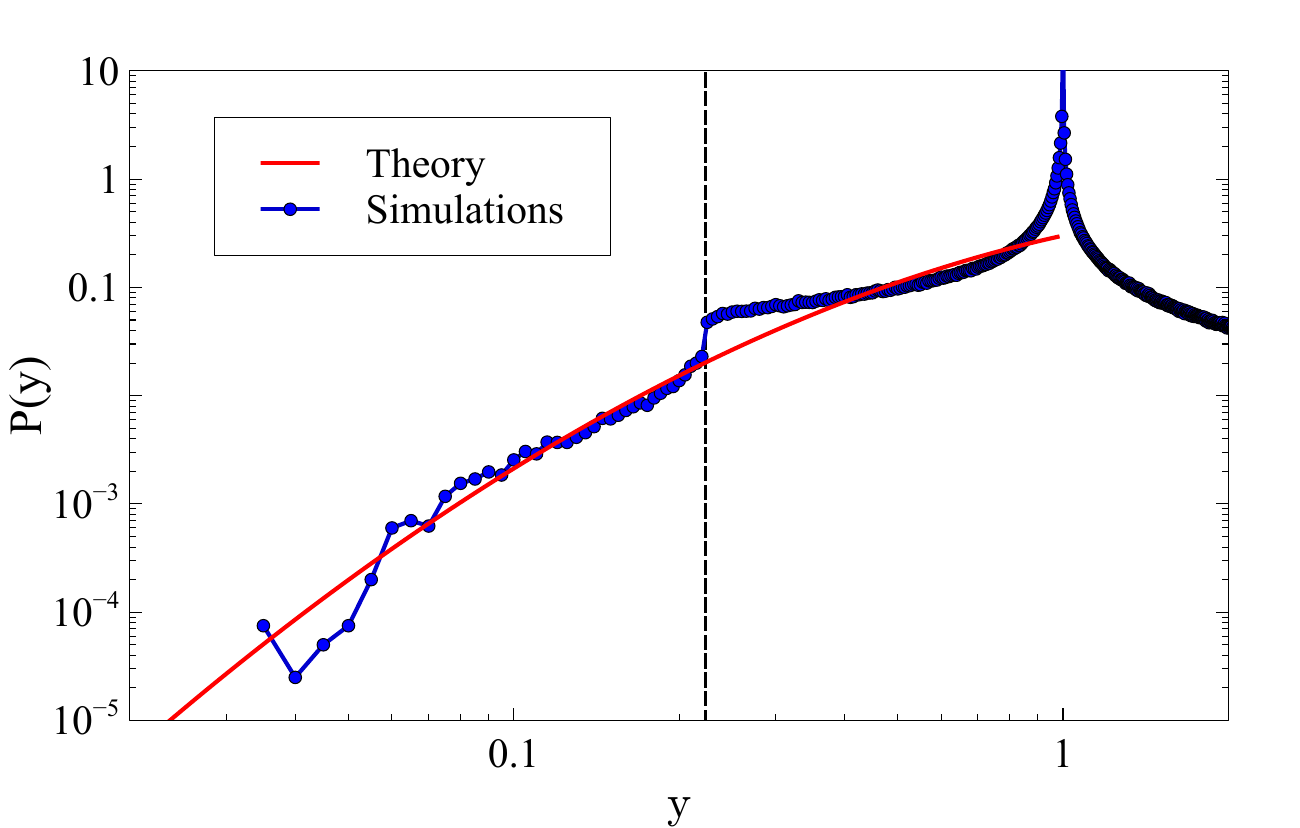}
    \caption{Probability density $P(y)$ from Monte Carlo integration of~\eqref{eq:sde} with the Marcus kick, 
    at $t=128$ (when stationarity is well established), compared with the asymptotic law
  \eqref{eq:wkb_result} for $z=0.5$ (left), $z=1.0$ (middle) and $z=1.5$ (right); 
  dichotomous jumps, $\beta=0.5$, $\gamma\tau=4$. The dashed line marks $y=e^{-z}$, below which
  Eq.~\eqref{eq:wkb_result} applies (see Sec.~\ref{sec:numerics}).}
    \label{fig:wkb}
\end{figure}

Numerical simulations were done with $\beta=0.5$ and $\gamma=1$, for the
four jump laws of Table~\ref{tab:crossing}, all with zero mean and
variance $a^2$, so that $z=a\beta$. Waiting times were drawn from
$\psi(\theta)=(\mu-1)T_0^{-1}(1+\theta/T_0)^{-\mu}$ with $T_0=1$ and
$\mu=1.5$ for all figures except Fig.~\ref{fig:wkb}, for which $\psi(\theta)=e^{-\theta/\tau}/\tau$ with $\gamma\tau=4$ was used.
Between events the drift
was integrated exactly, $x\mapsto x\,e^{-\gamma\theta}$. Trajectories,
of order $1.6\times10^{8}$ per case, started from a Gaussian of zero mean
and standard deviation $0.3$ and were followed to $t=128$
(to $t=256$ for Table~\ref{tab:heunacc}, whose averages are taken over $128\le t\le256$). 
The It\^o, Heun
and exact Marcus kicks were driven by the same random numbers (identical
event times and amplitudes), so that their differences are measured with
a variance much smaller than that of independent ensembles.
 

\subsection{The Heun kick integrator}

In Ref.~\refcite{BianucciMain} the SDE was integrated with Heun's scheme,
the standard algorithm for Stratonovich SDEs driven by a Wiener
process~\cite{Mannella2002,Mannella2025}. The Heun scheme for the delta kick integration was enforced as
\begin{equation}
  \tilde x = x-\xi I(x),
  \qquad
  x^{+}=x-\frac{\xi}{2}\left[I(x)+I(\tilde x)\right],
\end{equation}
which however reproduces \eqref{eq:expansion} through $O(\xi^2)$ 
only and therefore
implements Marcus not exactly but only to second order in the jump amplitude.

\subsection{The exact-flow kick integrator}

On the other hand, at 
a $\delta$ kick the
correct Marcus update is the exact flow map \eqref{eq:stratmap},
which for the cases of interest is available in closed form,
\begin{equation}
  I(x)=1 \ \Rightarrow\ \phi_{-\xi}(x)=x-\xi,
  \quad
  I(x)=1+\beta x \ \Rightarrow\
  \phi_{-\xi}(x)=-\tfrac1\beta+\left(x+\tfrac1\beta\right)e^{-\xi\beta},
\end{equation}
and otherwise requires one evaluation of $Y$ and one inversion.
      \begin{table}
      \begin{center}
      \begin{tabular}{cccc}
      \toprule
      $z=\xi\beta$ & $e^{-z}$ & Heun & rel.\ error \\
      \midrule
      0.25 & 0.7788 & 0.7813 & $+0.3\%$ \\
      0.50 & 0.6065 & 0.6250 & $+3.0\%$ \\
      1.00 & 0.3679 & 0.5000 & $+36\%$ \\
      2.00 & 0.1353 & 1.0000 & $+639\%$ \\
      3.00 & 0.0498 & 2.5000 & $+4921\%$ \\
      \bottomrule
      \end{tabular}
      \end{center}
      \caption{Contraction factor of the exact flow kick and of the Heun kick.}
      \label{tab:error}
      \end{table}
      
For $I=1+\beta x$ the Heun step acts as
$u\mapsto u\left(1-z+\tfrac{z^2}{2}\right)$ with $z=\xi\beta$, i.e.\ the
second-order Taylor truncation of the exact factor $e^{-z}$. This
has two consequences:
\begin{itemize}
\item the polynomial $1-z+z^2/2$ is the second-order truncation of
$e^{-z}$, and, like every truncation of even order, has no real root
(Sec.~\ref{sec:series}): the Heun kick preserves the barrier of
Proposition~\ref{prop:barrier} because of the parity of the truncation,
not because of its accuracy -- the third-order kick would cross it at
$z\simeq1.596$;
\item the magnitude is badly wrong once $z$ is not small, see
Table~\ref{tab:error}. The factor $1-z+z^2/2$ is smaller than one only for
$0<z<2$; for $z>2$ a kick that should contract towards the barrier pushes
the trajectory away from it.
\end{itemize}
This is a concrete, quantified reason to replace the Heun kick with the
exact flow map whenever $I$ is non-constant.

Figure~\ref{fig:wkb} compares the equilibrium density obtained with the Marcus kick and dichotomous 
jumps with the asymptotic law~\eqref{eq:wkb_result}. The data show a kink at $y=e^{-z}$, the image of 
the peak at $y=1$ under a single contracting kick; below it the density is built by repeated 
contracting kicks, which is exactly the regime of Eq.~\eqref{eq:wkb_result}. With $A$ in Eq.~\eqref{eq:wkb_result} fixed, for each $z$, by the 60th
percentile of the ratio between simulated and theoretical density over the
range $y<e^{-z}$, the agreement there is very good, confirming that under
the Marcus prescription the density vanishes faster than any power of $y$
as $y\to0$.


\begin{figure}
    \centering
    \includegraphics[width=0.32\linewidth]{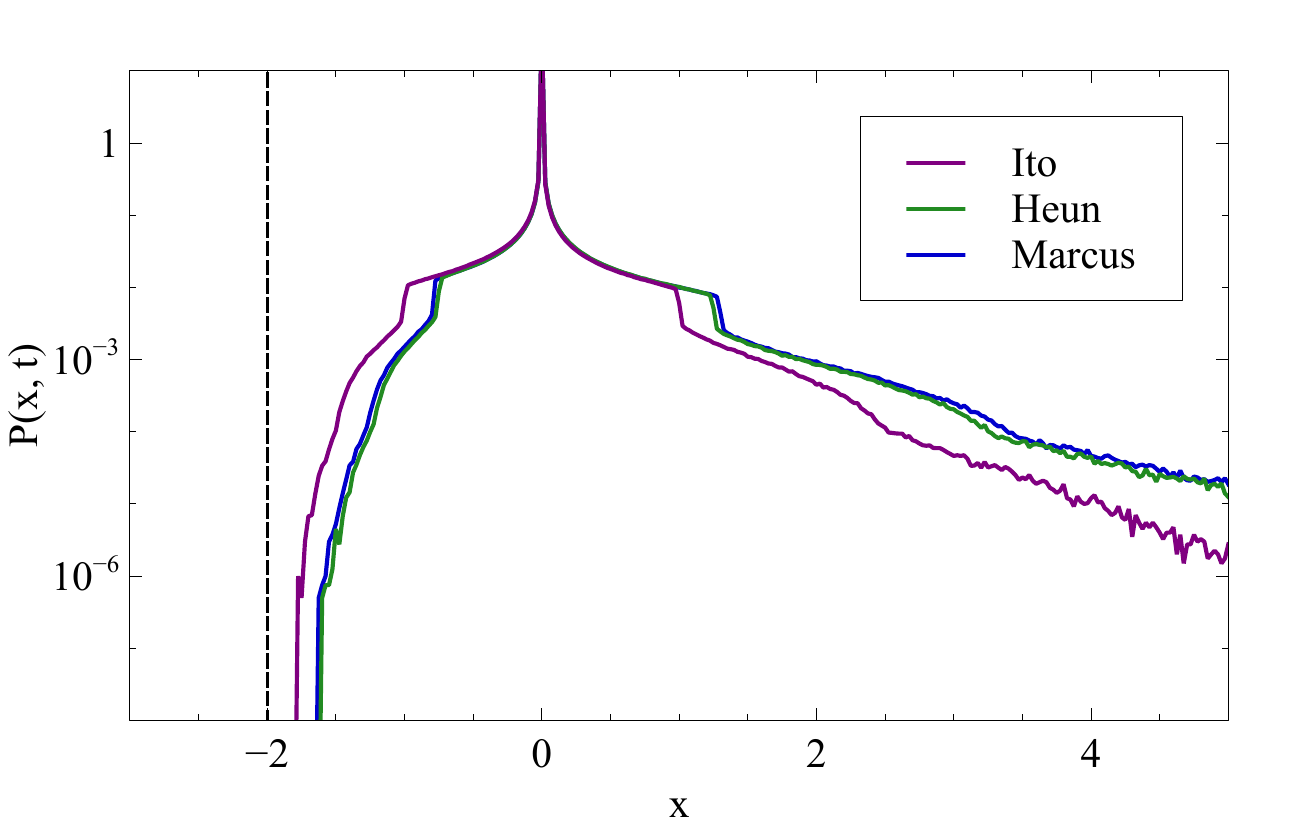}
    \includegraphics[width=0.32\linewidth]{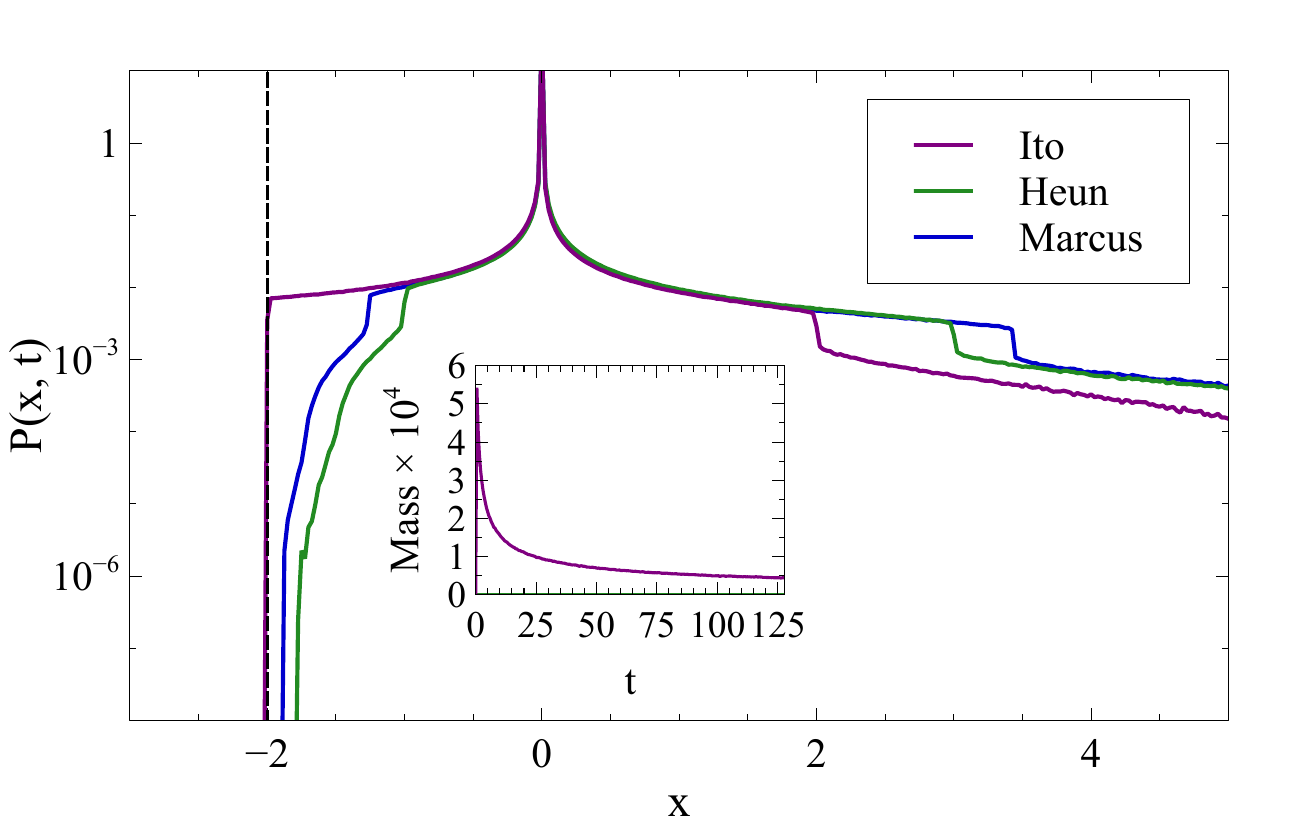}
    \includegraphics[width=0.32\linewidth]{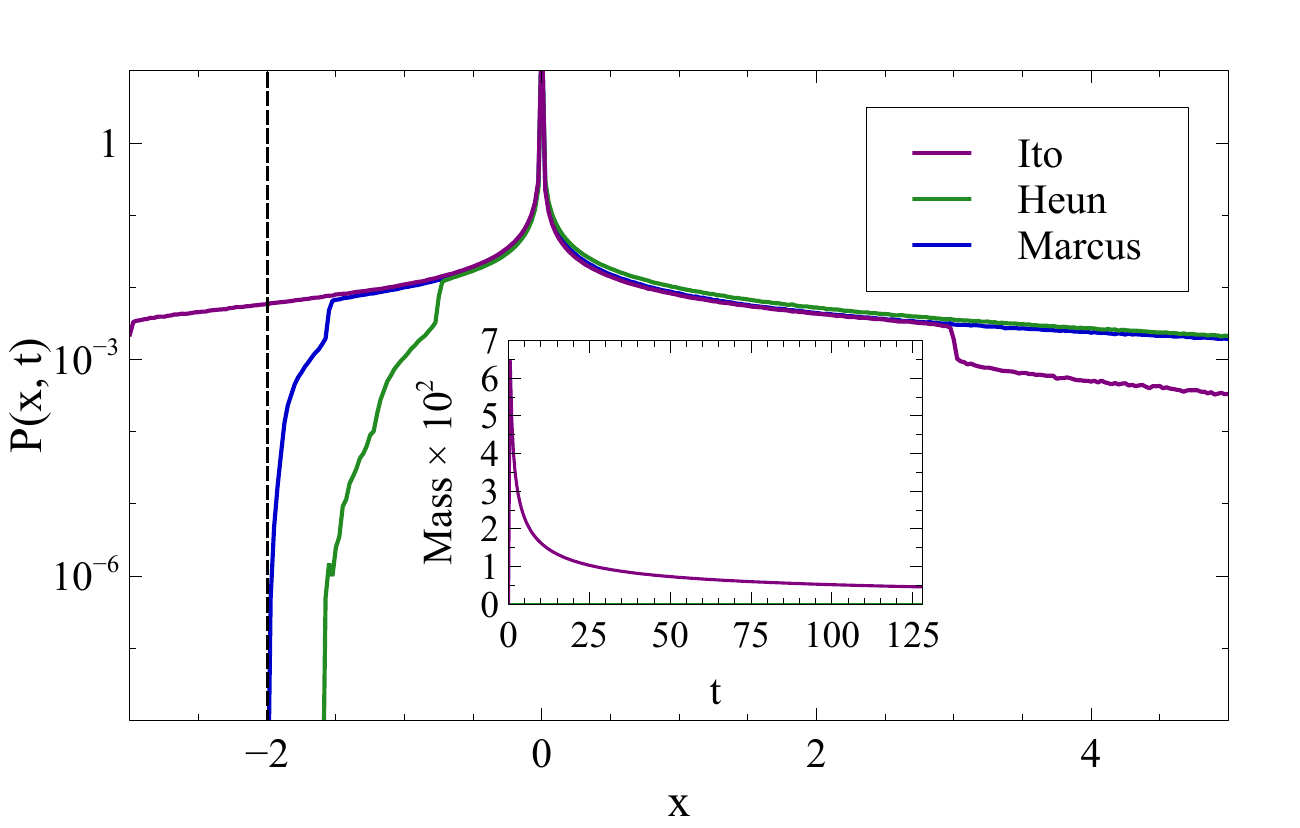}
    \caption{$P(x,t)$ at $t=128$ from Monte Carlo integration
  of \eqref{eq:sde}: 
  purple is It\^o, blue is Marcus, green is Heun, for
  $\beta=0.5$, dichotomous jumps. Left, $a\beta=0.5$ (below threshold), middle, $a\beta=1$ 
  (at threshold), right, $a\beta=1.5$ (above threshold). Only the It\^o density, and only in the right panel, 
  leaks below $x^*=-2$ (dashed line). 
  See the text for a description of the insets.}
    \label{fig1}
\end{figure}
\begin{figure}
    \centering
    \includegraphics[width=0.32\linewidth]{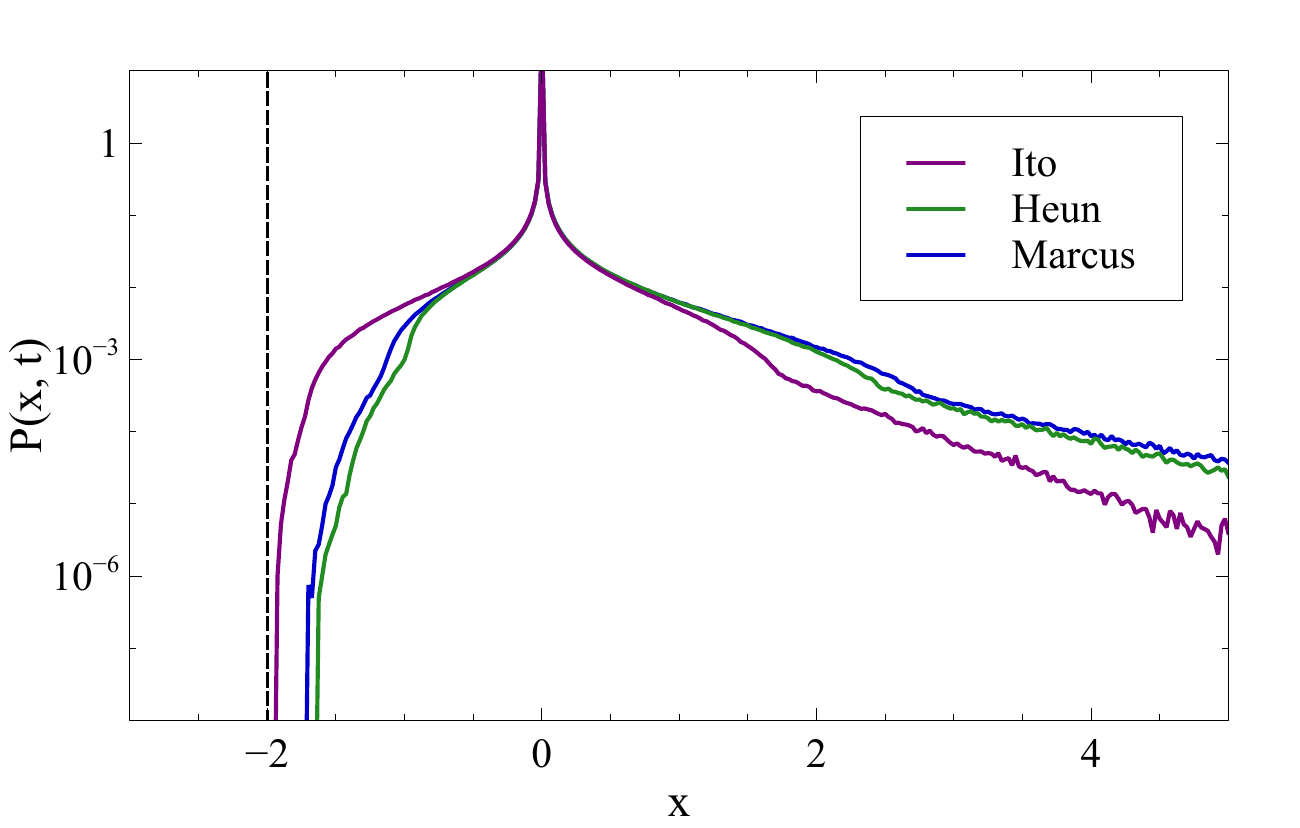}
    \includegraphics[width=0.32\linewidth]{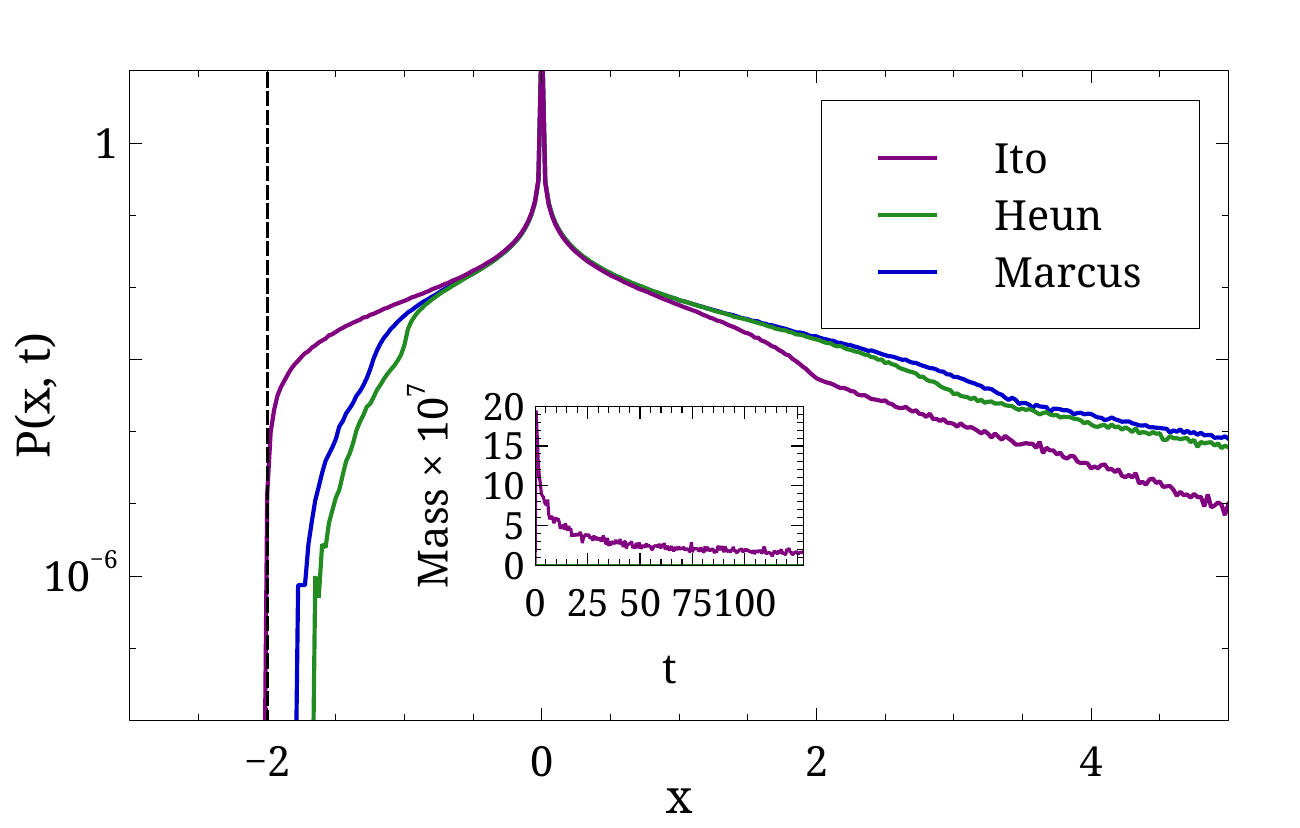}
    \includegraphics[width=0.32\linewidth]{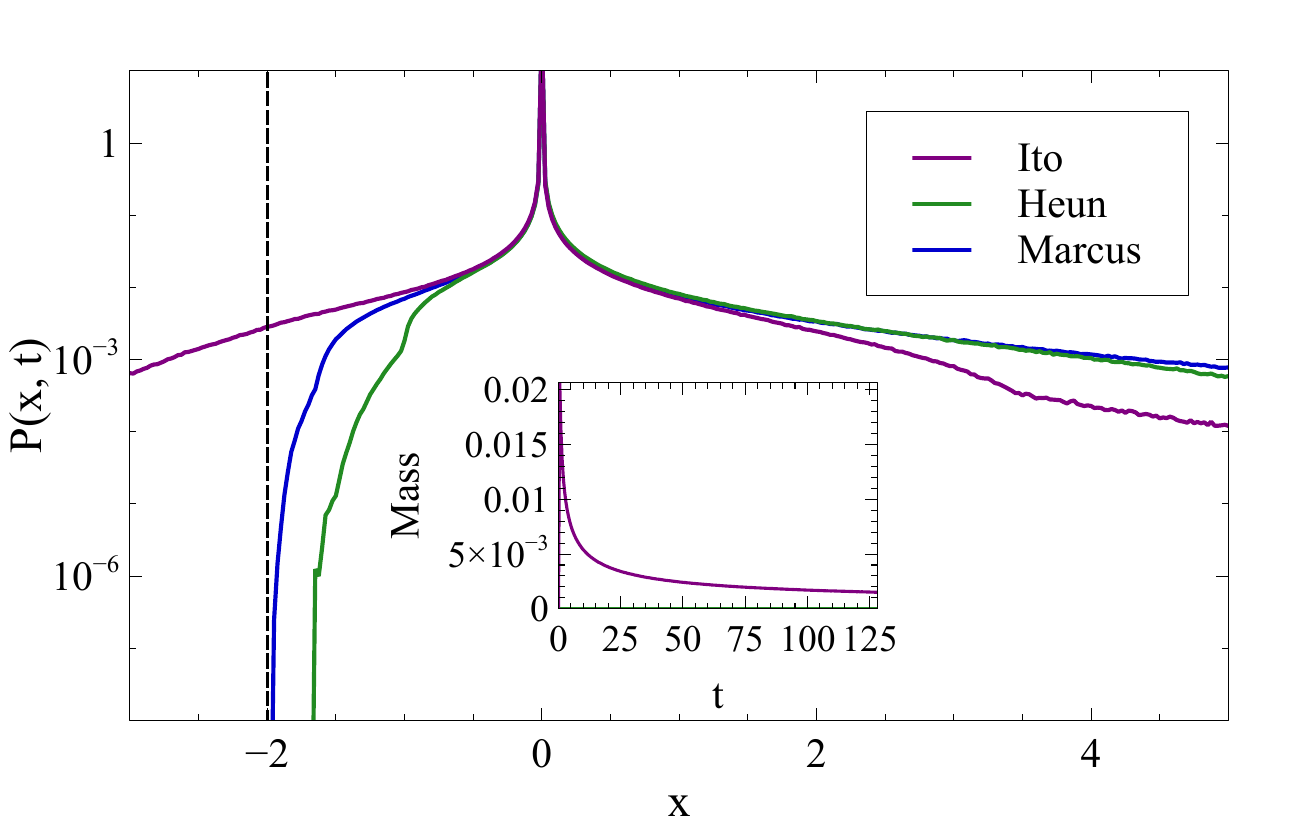}
    \caption{$P(x,t)$ at $t=128$ from Monte Carlo integration
  of \eqref{eq:sde}: 
  purple is It\^o, blue is Marcus, green is Heun, for
  $\beta=0.5$, uniform jumps. Left, $a\beta=0.5$ (below threshold), middle, $a\beta=1/\sqrt{3} \approx 0.577$ 
  (at threshold), right, $a\beta=1$ (above threshold). Only the It\^o density, and only in the right panel, leaks below $x^*=-2$ (dashed line). See the text for a description of the insets.}
    \label{fig2}
\end{figure}
\begin{figure}
    \centering
    \includegraphics[width=0.45\linewidth]{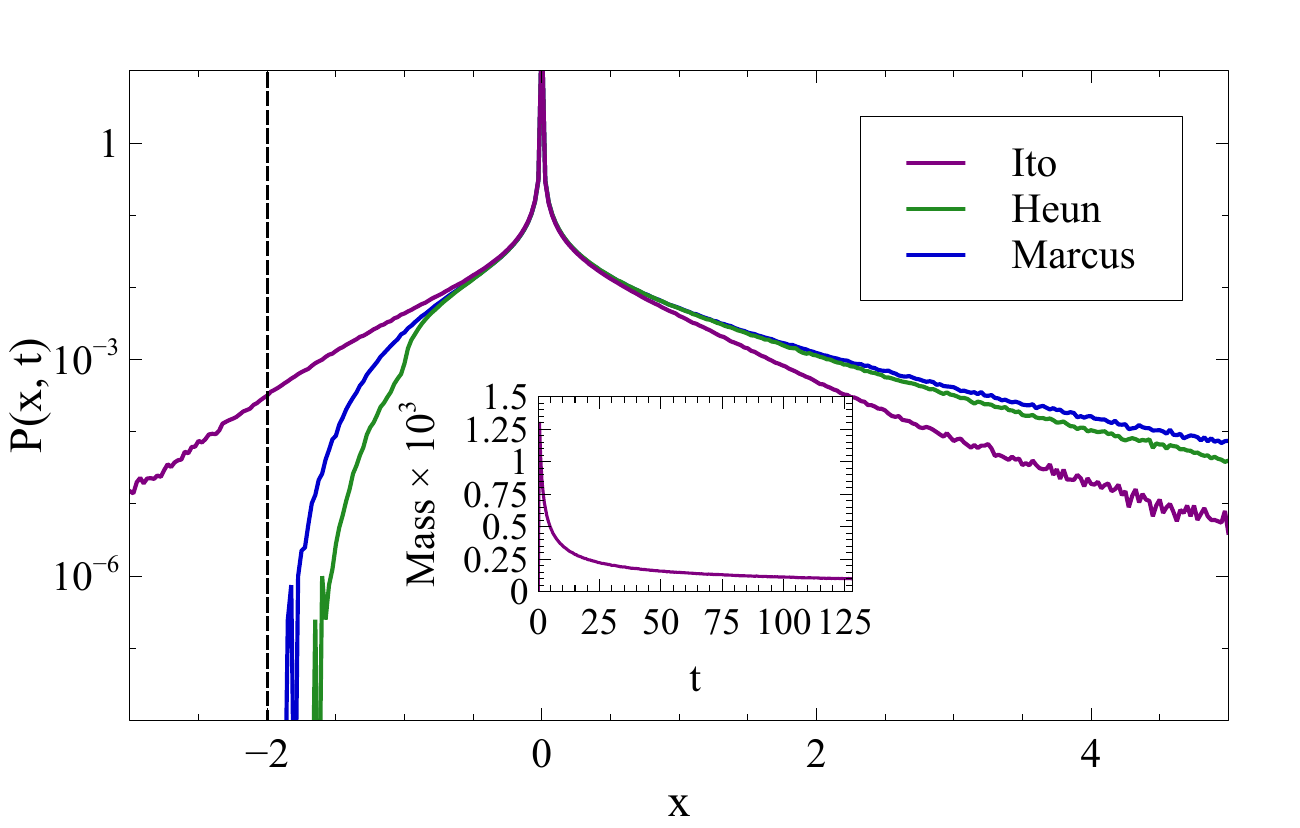}
    \includegraphics[width=0.45\linewidth]{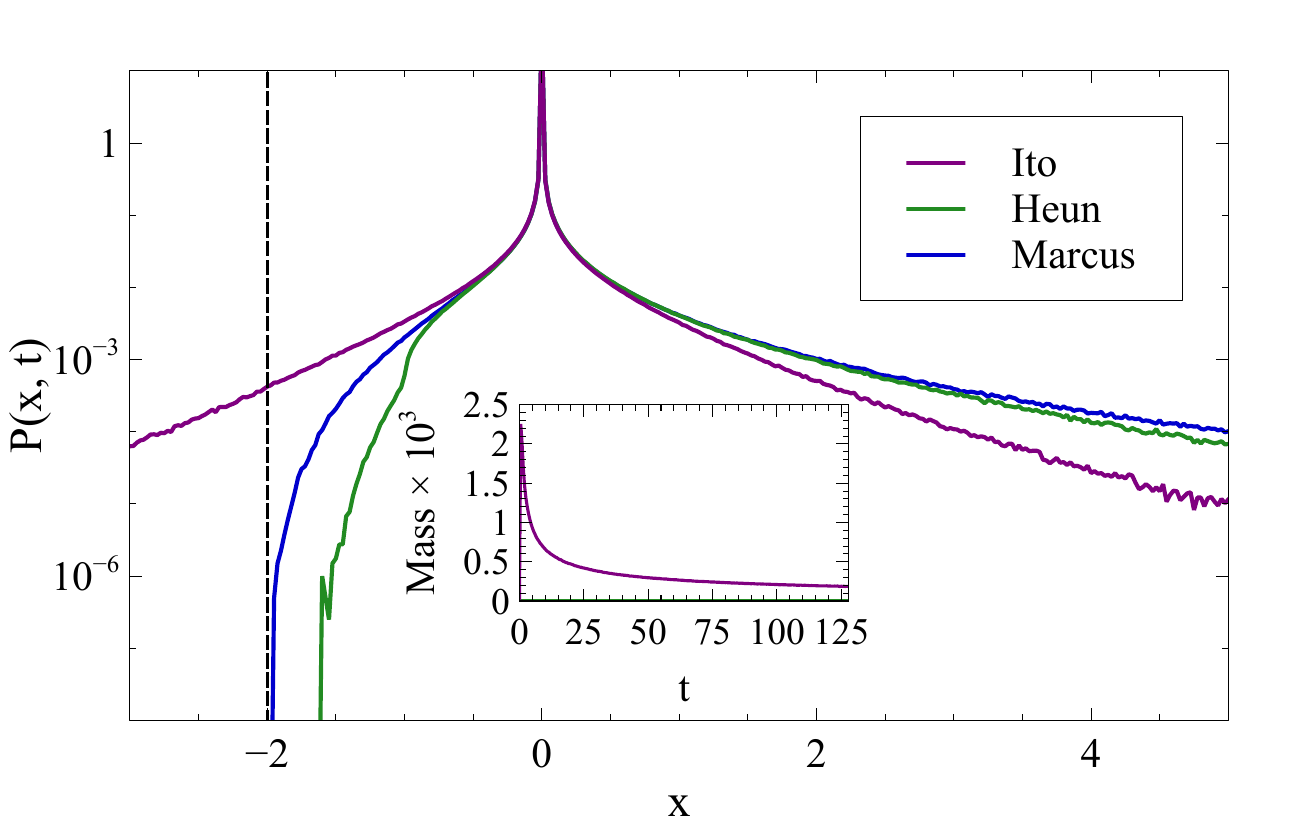}
    \caption{$P(x,t)$ at $t=128$ from Monte Carlo integration
  of \eqref{eq:sde}:  purple is It\^o, blue is Marcus, green is Heun, for
  $\beta=0.5$, $a=1$, and Gaussian jumps (left) or two-sided exponential jumps (right). Unlike the dichotomous case, the It\^o density leaks below $x^*=-2$ for any $a$, as expected for unbounded jumps. Insets: mass below $x^*$ as a function of time.}
    \label{fig3}
\end{figure}
\begin{figure}
    \centering
    \includegraphics[width=0.65\linewidth]{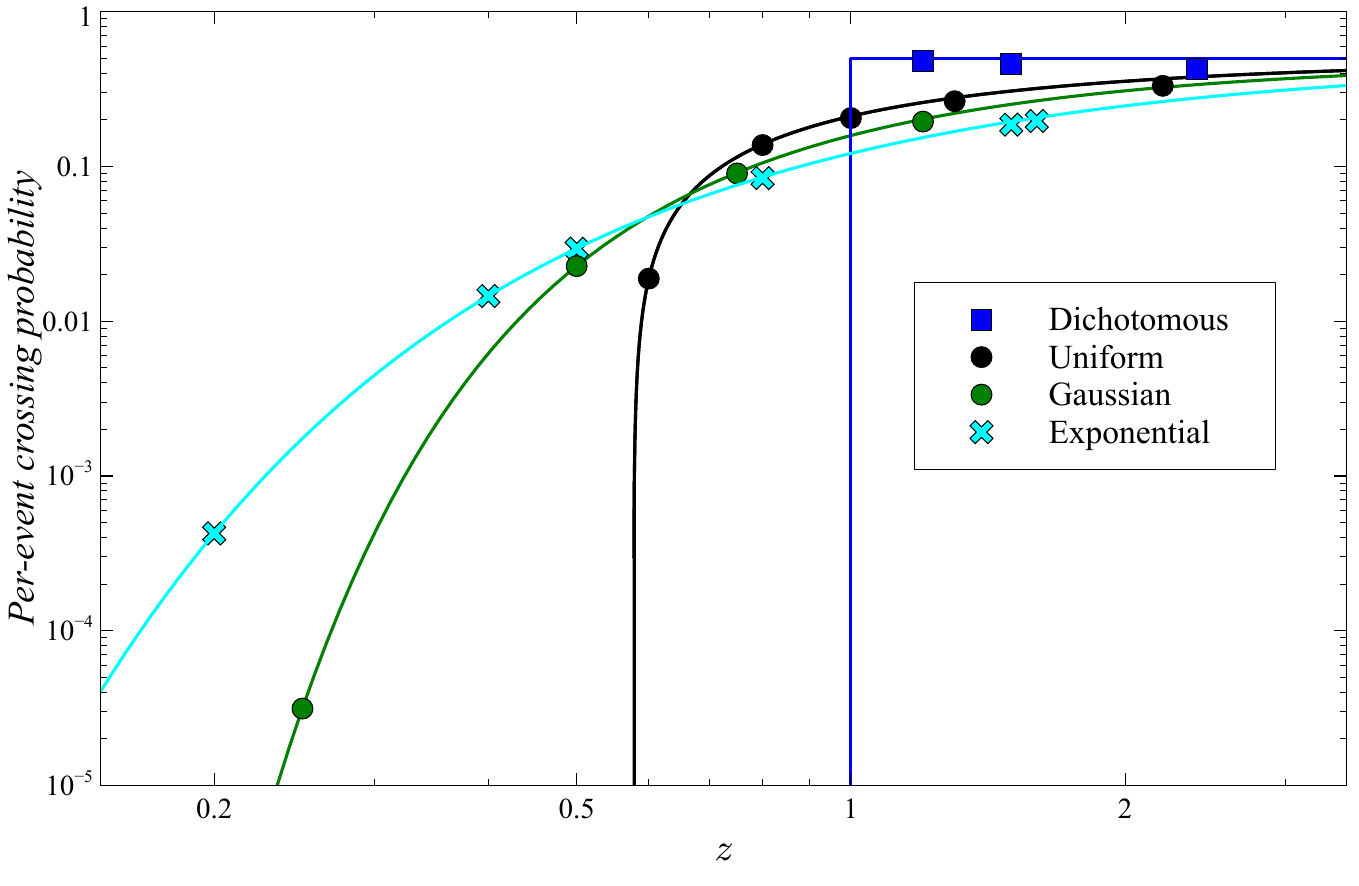}
    \caption{Per-event crossing probability versus $z$, for the four jump laws. 
    Symbols are the result of numerical simulations, 
    solid lines are the corresponding theoretical predictions 
    (see Table~\ref{tab:crossing}).}
    \label{figprob}
\end{figure}

Figures~\ref{fig1}--\ref{fig3} illustrate Proposition~\ref{prop:threshold}. Figure~\ref{fig1} shows $P(x,t)$ at $t=128$ 
for dichotomous jumps. 
For all values of $a$ 
considered, neither the Marcus nor the Heun kick leaks below $x^*=-2$. The It\^o kick does not leak either as long as $a\beta\le1$, 
but for $a\beta>1$ (right panel) it clearly does. The inset in the right panel shows the mass below $x^*$ as a function of time; 
it decreases because the drift carries trajectories back towards $x=0$. The inset in the middle panel shows the mass in the 
histogram bin containing $x^*$, i.e. within $dr/2=0.005$ of the barrier: these are It\^o trajectories just collapsed onto 
$x^*$ by kicks with $\xi\beta=1$ (Proposition 7(ii)) and carried back into $u>0$ by the drift. No mass lies strictly below 
$x^*$. Figure~\ref{fig2} shows the same behaviour for uniform jumps, except that at threshold the It\^o kick brings 
trajectories very close 
to $x^*$ (Proposition 7(i)) rather than exactly onto it. Figure~\ref{fig3} shows $P(x,t)$ at $t=128$ for Gaussian (left) 
and two-sided 
exponential (right) jumps. Although $a\beta<1$ in both cases, the It\^o density leaks markedly below $x^*$, as expected for 
unbounded jumps; again, neither the Marcus nor the Heun kick leaks.

Figure~\ref{figprob} compares the measured per-event crossing probabilities with the predictions of Table~\ref{tab:crossing}. The agreement is excellent. 

Figs.~\ref{fig1}--\ref{fig3} also show that the Heun density decays
towards the barrier faster than the exact Marcus one. The reason is not a
one-sided bias of the Heun factor ($1-z+z^2/2$ exceeds $e^{-z}$ for
$z>0$ but falls short of it for $z<0$) but the fact that it differs
from $e^{-z}$ by $z^3/6+O(z^4)$ with the sign of $z$. Every Heun jump is
therefore \emph{shorter} than the exact one, by a relative amount $z^2/6$,
in both directions. Shorter jumps produce a narrower density, and the
deficit is largest in the tails, where the mass is built by rare
sequences of large kicks. This is particularly evident in Fig.~\ref{fig1}, where the 
kinks in the distributions result from a single dichotomous kick applied to trajectories 
starting at the peak $x = 0$. Consequently, the distance between the kinks and $x = 0$ 
measures the size of the elementary kick under each scheme. In the left panel ($z = 0.5$), 
kinks are located at $x = \pm 1$ for the It\^o scheme, at $x = -0.79 / +1.30$ for Marcus,
and at $x = -0.75 / +1.25$ for Heun. In the middle panel, the It\^o distribution collapses
onto $x^\ast$ and features a kink at $x = 2$, whereas the kinks appear at 
$x = -1.26 /+3.44$ for Marcus and $x = -1.0 / +3.0$ for Heun. 
Finally, in the right panel, the It\^o
kinks are at $x = -3.0$ (outside the plot) and $x = +3.0$, the Marcus kinks are at $x =
-1.55$ and $x = 6.96$ (outside the plot), and the Heun kinks are at $x = -0.75$ and $x =
3.25$ (outside the plot).

\subsection{Accuracy of the Heun kick}
\label{sec:heunacc}
 
The shortening of the jumps has a simple quantitative consequence. Writing
$\zeta=\xi\beta$, the squared jump of $u$ is $u^2\zeta^2(1-\zeta+7\zeta^2/12)$
for the exact kick and $u^2\zeta^2(1-\zeta+\zeta^2/4)$ for Heun, so the
Heun jump variance is low by $u^2\overline{\zeta^4}/3$. Relative to the
leading term this is
\begin{equation}
  \frac{\Delta\overline{(\delta u)^2}}{\overline{(\delta u)^2}}
  \simeq -\frac{\overline{\zeta^4}}{3\,\overline{\zeta^2}}
  = -\frac{\kappa z^2}{3},
  \qquad
  \kappa=\frac{\overline{\xi^4}}{\overline{\xi^2}^{\,2}},
  \label{eq:heunbias}
\end{equation}
with $\kappa=1$, $1.8$, $3$ and $6$ for the dichotomous, uniform,
Gaussian and two-sided exponential laws. The error is $O(z^2)$ in
relative terms, one order lower than the $O(z^3)$ error of a single kick
in Table~\ref{tab:error}, and it grows with the kurtosis of the jump law.
Table~\ref{tab:heunacc} compares Eq.~\eqref{eq:heunbias} with the
measured bias of $\langle x^2\rangle$, together with the mass in the
region $x<-1$, half-way to the barrier.
 
\begin{table}[ht]
\begin{center}
{\begin{tabular}{@{}lcccc@{}}
\toprule
jump law & $z$ & $\Delta\langle x^2\rangle/\langle x^2\rangle$
  & $-\kappa z^2/3$ & $M_{\rm Heun}/M_{\rm exact}$\\
\midrule
dichotomous & 0.3 & $-2.99\%$  & $-3.0\%$  & $0.750$ \\
dichotomous & 0.5 & $-8.34\%$  & $-8.3\%$  & $0.719$ \\
dichotomous & 1.0 & $-44.3\%^{\rm a}$ & $-33.3\%$ & $0.214$ \\
Gaussian    & 0.3 & $-8.98\%$  & $-9.0\%$  & $0.366$ \\
Gaussian    & 0.5 & $-24.3\%$  & $-25.0\%$  & $0.211$ \\
Gaussian    & 1.0 & $-80.9\%^{\rm a}$ & -- & $0.081$ \\
Uniform    & 0.3 & $-5.32\%$  & $-5.4\%$  & $0.667$ \\
Uniform    & 0.5 & $-14.6\%$  & $-15.0\%$  & $0.377$ \\
Uniform    & 1.0 & $-61.7\%^{\rm a}$ & -- & $0.100$ \\
two-sided exp.    & 0.3 & $-17.1\%$ & $-18.0\%$ & $0.172$ \\
two-sided exp.    & 0.5 & $-46.1\%^{\rm a}$ & $-50.0\%$ & $0.131$ \\
two-sided exp.    & 1.0 & $-99.7\%^{\rm b}$ & -- & $0.074$ \\
\bottomrule
\end{tabular}}
\caption{Heun kick versus exact Marcus kick, $\beta=0.5$, $\mu=1.5$, same noise realisations; second moments and masses averaged over $128\le t\le256$. Relative bias of $\langle x^2\rangle$, measured and from Eq.~\eqref{eq:heunbias} (leading order; “—” where it does not apply), and ratio of the probability masses in $x<-1$ (Heun/exact). $^{\rm a}$\,Estimate dominated by rare large excursions, because the kick
factor has a large or infinite fourth moment: independent runs differ by
several percent.
$^{\rm b}$\,The exact $\langle x^2\rangle$ is infinite, since
$\langle e^{2\beta\xi}\rangle$ diverges for this law when $z\ge1/\sqrt2$.}
\label{tab:heunacc}
\end{center}
\end{table}

Two features stand out. First, Eq.~\eqref{eq:heunbias} is accurate up to $z\approx0.5$ for the dichotomous, uniform and Gaussian laws, so the bias of the moments can be predicted before any simulation is run; for the two-sided exponential law the estimate at $z=0.5$ is not converged, for the reason given in the caption of Table~\ref{tab:heunacc}. Second, the near-barrier region -- the
subject of this paper -- is where the Heun kick is worst: for Gaussian
jumps at $z=0.5$ it underestimates the mass in $x<-1$ by a factor of more
than four, while the per-kick error of Table~\ref{tab:error} is only 3\%.
The simulations of Ref.~\refcite{BianucciMain} used dichotomous jumps at
$z=0.5$, where the second-moment bias is about 8\%.

\section{Which prescription does an experiment realise?}

The answer is set by the comparison of two time scales and one amplitude.
Let $\epsilon$ be the true duration of the physical impulse and
$T_{\mathrm{rel}}$ the relaxation time of $x$ under the drift. If
$\epsilon\ll T_{\mathrm{rel}}$ the spike idealisation is justified, and
Proposition~\ref{prop:shape} then forces the Marcus prescription
\emph{regardless of how small $\epsilon$ is}: the limit is singular in
$\epsilon$ but not in the prescription. The It\^o prescription is appropriate
only when the event is genuinely atomic and its effect is specified
externally -- for instance a counter increment, a discrete packet arrival
whose size is fixed before the system state is consulted, or a model in
which $I(x)$ represents a measured gain sampled before the event. In
other words:
\begin{itemize}
\item physical impulse of short but nonzero duration, gain varying
      continuously with the state $\Rightarrow$ Marcus;
\item discrete, externally specified event whose magnitude is determined
      by the pre-event state $\Rightarrow$ It\^o.
\end{itemize}
A concrete instance of the first class is conductance-based synaptic
input. A synaptic event of strength $a$ opens a brief conductance pulse,
$\dot V=g(t)\,(E-V)$ with $g(t)=a\,\delta(t-t_q)$, so that the gain
$I(V)\propto E-V$ vanishes at the reversal potential $E$. The flow-map
kick $V\mapsto E+(V-E)\,e^{-a}$ is the exact solution during a short
conductance pulse~\cite{Richardson2005}, and the membrane potential
approaches its reversal potential without ever crossing it. The It\^o
kick $V\mapsto V+a(E-V)$ overshoots the reversal potential whenever
$a>1$, which is physically impossible: here the physics settles the
prescription, and it is Marcus.
 
When the pulse duration and the relaxation time of the system are
comparable, neither idealisation applies. For systems with inertia driven
by compound Poisson coloured noise, Li, Min and Wang showed that adiabatic
elimination yields the It\^o or the Marcus form, or an intermediate one,
depending on the ratio between the relaxation time and the correlation
time of the noise~\cite{AdiabaticPRE2014}.

Where $I$ has a zero, the two are not interchangeable at any order, and
the choice must be made on physical grounds before any fitting is done.

\section{Conclusions}

In contrast to Gaussian white noise, where the It\^o and Stratonovich prescriptions differ by $O(\dd t)$ per
time step and the difference accumulates into a finite spurious drift, state-dependent renewal shot noise with finite random jump amplitudes $\xi_q$ produces a finite, non-negligible discrepancy at every single kick event. This discrepancy cannot be re-absorbed into a renormalized drift operator $\Delta C(x)$ beyond lowest order in the jump amplitude. Every term of that series vanishes where $I$ does, so a renormalised drift
cannot alter the transport across a zero of $I$, whereas the It\^o kick
can carry probability across it. Consequently, wherever $I$ vanishes the choice between It\^o and Marcus
is non-perturbative: for linear $I$, below a threshold it amounts to a
remapping of the jump law, above it to a change in the topology of the
dynamics.

Structurally, the Marcus kick represents the exact time-$(-\xi_q)$ map of the autonomous flow $dx/ds = I(x)$. It acts as a diffeomorphism on the real line that preserves each maximal sign interval of $I(x)$ and maintains the group property $\phi_{s_1} \circ \phi_{s_2} = \phi_{s_1 + s_2}$. Zeros $x^\ast$ of $I(x)$ therefore cannot be crossed by any Marcus kick,
whatever the jump amplitudes or waiting-time statistics. Combined with a
drift directed into the physical sector, as for $C(x)=\gamma x$, they act
as one-way barriers;  in the stationary state, computed here for exponential waiting 
times, the density vanishes there faster than any power of $1+\beta x$, $P_{\rm eq}(x^*)=0$. In contrast, the It\^{o} kick represents only the tangent approximation
$m_\xi(x)$ of this flow and lacks both invertibility and the group structure. For gain functions such as $I(x) = 1 + \beta x$, the It\^{o} prescription exhibits a sharp
threshold at $\xi\beta=1$: below threshold it preserves orientation, at threshold it collapses the entire line onto $x^*$, and above threshold it reverses orientation,
allowing trajectories to cross $x^*$. As a result, for unbounded jump distributions (such as Gaussian or exponential), It\^{o} dynamics leaks across the barrier with finite
probability per event across all parameter ranges.

An exact non-local master equation exists under either prescription, but
only the Marcus kick is generated by a single operator; the closed-form
master equation of Ref.~\refcite{BianucciMain}, whose generator depends on
the jump law only through its characteristic function, is therefore
intrinsically a Marcus equation. Finally, from a numerical perspective, standard white-noise integration schemes such as the Heun kick method only match the Marcus flow to second order $O(\xi^2)$. While Heun accidentally preserves the barrier for $I(x)=1+\beta x$, because it is a truncation of even order, 
it shortens every jump:
second moments are biased by a relative $\kappa z^2/3$, and the mass near
the barrier is depleted already at moderate $z$ -- by a factor of four for
Gaussian jumps at $z=0.5$. To integrate stochastic differential equations driven by state-dependent shot noise accurately without artificial numerical artifacts, algorithms must employ the exact flow map $\phi_{-\xi}(x)$ in place of traditional white-noise approximations.

\section*{Acknowledgments}
We thank the Green Data Center of University of Pisa for providing the computational power needed for the present paper. During the preparation of this manuscript, the authors used Claude (https://claude.ai) for the purposes of improving the English and readability of some of the textual content. The authors have reviewed and edited the output and take full responsibility for the content of this publication. Data are available upon reasonable request.

\end{document}